\documentclass[letterpaper,11pt]{article}
\usepackage[margin=1in]{geometry}
\usepackage{amsfonts,amsmath,amssymb,amsthm}
\usepackage{cite,comment,enumerate,hyperref}
\usepackage[colorinlistoftodos,textsize=small,color=red!25!white,obeyFinal]{todonotes}

\newtheorem{theorem}{Theorem}[section]

\newtheorem{lemma}[theorem]{Lemma}

\newtheorem{innercustomthm}{Theorem}
\newenvironment{customthm}[1]
  {\renewcommand\theinnercustomthm{#1}\innercustomthm}
  {\endinnercustomthm}

\theoremstyle{definition}
\newtheorem{definition}{Definition}

\theoremstyle{remark}

\usepackage{xcolor,xspace}
\usepackage{booktabs}

\usepackage{algorithm}
\usepackage[noend]{algorithmic}

\usepackage{csquotes}

\newcommand{\poly}{\text{poly}}
\newcommand{\test}{\textsc{test}}
\newcommand{\dist}{d}
\newcommand{\Oish}{\widetilde{O}}

\newcommand{\Erdos}{Erd{\" o}s}

\DeclareMathOperator*{\E}{\mathbb{E}}

\title{Optimal Vertex Fault-Tolerant Spanners in Polynomial Time}
\author{Greg Bodwin\thanks{Supported in part by NSF awards CCF-1717349,  DMS-183932 and CCF-1909756.}\\University of Michigan\\ \texttt{bodwin@umich.edu} \and Michael Dinitz\thanks{Supported in part by NSF award CCF-1909111.}\\Johns Hopkins University\\ \texttt{mdinitz@cs.jhu.edu} \and Caleb Robelle\\UMBC\\ \texttt{carobel1@umbc.edu}}
\date{}

\begin{document}

\begin{titlepage}

\maketitle

\begin{abstract}
    Recent work has pinned down the existentially optimal size bounds for vertex fault-tolerant spanners: for any positive integer $k$, every $n$-node graph has a $(2k-1)$-spanner on $O(f^{1-1/k} n^{1+1/k})$ edges resilient to $f$ vertex faults, and there are examples of input graphs on which this bound cannot be improved.
    However, these proofs work by analyzing the output spanner of a certain exponential-time greedy algorithm.
    In this work, we give the first algorithm that produces vertex fault tolerant spanners of optimal size and which runs in polynomial time.
    Specifically, we give a randomized algorithm which takes $\Oish\left( f^{1-1/k} n^{2+1/k}  + mf^2\right)$ time.  We also derandomize our algorithm to give a deterministic algorithm with similar bounds.  
    This reflects an exponential improvement in runtime over [Bodwin-Patel PODC '19], the only previously known algorithm for constructing optimal vertex fault-tolerant spanners.
\end{abstract}
\thispagestyle{empty}
\end{titlepage}

\section{Introduction}

Let $G = (V, E)$ be a graph, possibly with edge lengths $w : E \rightarrow \mathbb{R}_{\geq 0}$. A $t$-spanner of $G$, for $t \geq 1$, is a subgraph $H = (V, E')$ that preserves all pairwise distances within a factor of $t$, i.e.,
\begin{align}
d_{H}(u,v) \leq t \cdot d_G(u,v) \label{eq:nonfttest}
\end{align}
for all $u,v \in V$ (where $d_{X}$ denotes the shortest-path distance in a graph $X$).  Since $H$ is a subgraph of $G$ it is also true that $d_G(u,v) \leq d_{H}(u,v)$, and so distances in $H$ are the same as in $G$ up to a factor of $t$.  The distance preservation factor $t$ is called the \emph{stretch} of the spanner.  
Spanners were introduced by Peleg and Ullman~\cite{PelegU:89} and Peleg and Sch{\"{a}}ffer~\cite{PelegS:89}, and have a wide range of applications in routing \cite{PelegU:89-routing}, synchronizers \cite{awerbuch1990network}, broadcasting \cite{awerbuch1991cient,peleg2000distributed}, distance oracles \cite{thorup2005approximate}, graph sparsifiers \cite{kapralov2012spectral}, preconditioning of linear systems \cite{elkin2008lower}, etc.

The most common objective in spanners research is to achieve the best possible existential size-stretch trade-off, and to do this with algorithms that are as fast as possible.
Most notably, a landmark result of Alth\"ofer et al.~\cite{AlthoferDDJS:93} analyzed the following simple and natural greedy algorithm: given an $n$-node graph $G$ and an integer $k \geq 1$, consider the edges of $G$ in non-decreasing order of their weight and add an edge $(u,v)$ to the current spanner $H$ if and only if $\dist_{H}(u,v)>(2k-1) w(u,v)$.
They proved that this algorithm produces $(2k-1)$-spanners of \emph{existentially optimal size}: the spanner produced has size $O(n^{1+1/k})$, and (assuming the well-known \emph{Erd\H{o}s girth conjecture}~\cite{erdHos1964extremal}) there are graphs in which \emph{every} $(2k-1)$ spanner (and in fact every $2k$-spanner) has at least $\Omega(n^{1+1/k})$ edges.  

\subsection{Fault Tolerance}

A crucial aspect of real-life systems that is not captured by the standard notion of spanners is the possibility of \emph{failure}.
If some edges (e.g., communication links) or vertices (e.g., computer processors) fail, what remains of the spanner might not still approximate the distances of what remains of the original graph.
This motivates the notion of \emph{fault tolerant} spanners: 
\begin{definition} \label{def:FT}
A subgraph $H$ is an $f$-vertex fault tolerant ($f$-VFT) $t$-spanner of $G$ if
\begin{align}
d_{H \setminus F}(u,v) \leq t \cdot d_{G \setminus F}(u,v) \label{eq:fttest}
\end{align}
for all $u,v \in V$ and $F \subseteq V \setminus \{u,v\}$ with $|F| \leq f$.
\end{definition}
In other words, an $f$-VFT spanner contains a spanner of $G\setminus F$ for every set of $|F| \le f$ nodes that could fail.  The definition for edge fault tolerance (EFT) is equivalent, with the only change being that $F \subseteq E$ rather than $F \subseteq V \setminus \{u,v\}$.

Fault tolerant spanners were originally introduced in the geometric setting (where the vertices are points in $\mathbb{R}^d$ and the initial graph $G$ is the complete graph with Euclidean distances) by Levcopoulos, Narasimhan, and Smid \cite{levcopoulos1998efficient} and have since been studied extensively in that setting~\cite{LNS98,lukovszki1999new,czumaj2004fault,NS07}. 
Chechik, Langberg, Peleg and Roditty \cite{ChechikLPR:10} were the first to study fault-tolerant spanners in general graphs, giving a construction of an $f$-VFT $(2k-1)$-spanner of size approximately $O(f^2 k^{f+1} \cdot n^{1+1/k}\log^{1-1/k}n)$ and an
$f$-EFT $(2k-1)$-spanner of size $O(f\cdot n^{1+1/k})$.  So they showed that introducing tolerance to $f$ edge faults costs us an extra factor of $f$ in the size of the spanner, while introducing tolerance to $f$ vertex faults costs us a factor of $f^2 k^{f+1}$ in the size (compared to the size of a non-fault tolerant spanner of the same stretch).  

\renewcommand{\arraystretch}{1.5} 

\begin{table*}[t]
\begin{center}

    \begin{tabular}{llcl}
    \toprule
        \textbf{Spanner Size} & \textbf{Runtime} & \textbf{Greedy?} & \textbf{Citation} \\
    \midrule
        $\Oish \left(k^{O(f)} \cdot n^{1+1/k} \right)$  & $\Oish \left(k^{O(f)} \cdot n^{3+1/k} \right)$ & & \cite{ChechikLPR:10} \\
        $\Oish \left(f^{2 - 1/k} \cdot n^{1+1/k} \right)$ & {\color{blue} $\Oish\left( f^{2 - 2/k} \cdot m n^{1+1/k} \right)$} & & \cite{DinitzK:11}\\
        $O \left( \exp(k) f^{1 - 1/k} \cdot n^{1+1/k} \right) $ & $O\left( \exp(k) \cdot m n^{O(f)} \right)$ & \checkmark{} & \cite{BDPW18} \\
        \color{red} $O \left( f^{1 - 1/k} \cdot n^{1+1/k} \right)$ & $O\left( mn^{O(f)} \right)$ & \checkmark{} & \cite{BP19} \\
        $O \left( k f^{1 - 1/k} \cdot n^{1+1/k} \right)$ & \color{blue}{$\Oish\left( f^{2 - 1/k} \cdot m n^{1+1/k} \right)$} & (\checkmark{}) & \cite{DR20} \\
        \color{red} $O \left( f^{1 - 1/k} \cdot n^{1+1/k} \right)$ & \color{blue}$\Oish\left( f^{1 - 1/k} n^{2+1/k} + mf^2\right)$ & (\checkmark{}) & \textbf{(this paper)}\\
    \bottomrule
    \end{tabular}
    
    \caption{\label{tab:priorwork} Prior work on $f$-VFT $(2k-1)$-spanners of weighted input graphs on $n$ nodes and $m$ edges.  Size bounds in red are existentially optimal, and runtimes in blue are polynomial.  The (\checkmark{}) entries indicate a greedy algorithm with slack, as discussed below.  With $\Oish$ we hide factors of $\log n$ (or $k$, since we may assume $k \le \log n$).}
    \end{center}

\end{table*}

Since~\cite{ChechikLPR:10}, there has been a line of work focused on improving these bounds, particularly for vertex faults (see Table \ref{tab:priorwork}).  The first improvement was by~\cite{DinitzK:11}, who improved the bound for vertex faults to $O(f^{2-\frac{1}{k}}n^{1+\frac{1}{k}}\log{n})$ via a black-box reduction to non-fault tolerant spanners.
Following this, the area turned towards analyses of the \emph{FT-greedy algorithm}, the obvious extension of the greedy algorithm of~\cite{AlthoferDDJS:93} to the fault tolerant setting: 
look at the edges $(u, v)$ of the input graph $G$ in order of nondecreasing weight, and add $(u, v)$ to the spanner $H$ iff currently there exists a set of $|F| \le f$ faults such that (\ref{eq:fttest}) fails.
This algorithm was first analyzed by~\cite{BDPW18} who obtained a size bound of $O(\exp(k) f^{1-1/k} n^{1+1/k})$.  They also proved a lower bound of $\Omega(f^{1-1/k} n^{1+1/k})$ for $f$-VFT $(2k-1)$-spanners: assuming the girth conjecture of \Erdos{} \cite{erdHos1964extremal}, there are graphs which require that many edges for any $f$-VFT $(2k-1)$-spanner.
An improved analysis of the FT-greedy algorithm was then given by~\cite{BP19}, who remove the $\exp(k)$ factor and so proved that this algorithm gives existentially optimal VFT spanners. 

While the FT-greedy algorithm inherits some of the nice properties of the non-faulty greedy algorithm (such as simplicity, easy-to-prove correctness, and existential optimality), it unfortunately has serious issues in runtime.
The edge test in the FT greedy algorithm, i.e., whether or not there exists a fault set under which (\ref{eq:fttest}) holds, is an NP-hard problem known as \textsc{length-bounded cut} \cite{BEHKSS06}, and hence the algorithm inherently runs in exponential time.
Addressing this, a \emph{greedy algorithm with slack} was recently proposed in~\cite{DR20}.
This algorithm is an adaptation of the FT-greedy algorithm which replaces the exponential-time edge test with a different subroutine $\test(u, v)$, which accepts every edge $(u, v)$ where there exist $|F| \le f$ faults under which (\ref{eq:fttest}) fails, \emph{and possibly some other edges too}.
This slack maintains correctness and allows one to escape NP-hardness, but it introduces the challenge of bounding the number of additional edges added.
The approach in \cite{DR20} is to design an $O(k)$-approximation algorithm for \textsc{length-bounded cut} and use this in an efficiently computable $\test$ subroutine.
This gives a polynomial runtime, but pays the approximation ratio of $O(k)$ in spanner size over optimal.
So the result in~\cite{DR20} takes an important step forward (polynomial time) but also a step back (non-optimal size, by a factor of $O(k)$).

It thus remains an important open problem to design a polynomial time algorithm which obtains truly optimal size.
We note that $k$ factors are often considered particularly important for spanners, since the regime $k = \Theta(\log n)$ yields the sparsest possible spanners and hence arises commonly in algorithmic applications (see, e.g., \cite{BBGNSSS20} for a recent example), and here an extra factor of $k$ in the size of the spanner is significant.
Accordingly, for spanners and many related objects there has been significant effort expended to remove unnecessary factors of $k$.
It seems to often be the case that initial algorithms pay a factor of $k$, which can later be removed through more careful algorithms and analyses.
Our results fit into this tradition, addressing the remaining open question: can we get truly optimal-size fault tolerant spanners in polynomial time?

\subsection{Our Results and Techniques}
 
We answer this in the affirmative, giving both randomized and deterministic algorithms for constructing optimal-size fault-tolerant spanners.  More formally, we prove the following theorems.
\begin{theorem} \label{thm:randomized-main}
There is a randomized algorithm which runs in expected time
$$\Oish\left(f^{1-1/k} n^{2+1/k} + mf^2 \right)$$
which with high probability returns an $f$-VFT $(2k-1)$-spanner with $O(f^{1-1/k} n^{1+1/k})$ edges.
\end{theorem}

\begin{theorem} \label{thm:deterministic-main}
There is a deterministic algorithm which constructs an $f$-VFT $(2k-1)$ spanner with at most $O(f^{1-1/k} n^{1+1/k})$ edges in time 
\[\tilde O\left(f^{4-1/k} n^{2+1/k} + mf^5\right).\] 
If $f = \poly(n)$ (i.e., $f \geq n^c$ for some constant $c > 0$) then the running time improves to 
\[\tilde O\left(f^{1-1/k} n^{2+1/k} + mf^2\right),\] 
where the polynomial exponent $c$ appears in both the spanner size and the running time but is hidden by the $O(\cdot)$ notation. 
\end{theorem}

So if $f$ is subpolynomial in $n$ then our deterministic algorithm is slower than our randomized algorithm by about $f^3$, while if $f$ is polynomial in $n$ then we get determinism essentially for free (although the hidden polylogarithmic factors are larger in the deterministic case).  Recent work by Karthik and Parter~\cite{karthik2021deterministic} provides a better derandomization which gives the same $\tilde O\left(f^{1-1/k} n^{2+1/k} + mf^2\right)$ bound for all values of $f$.

To put these results in context, note that this is an \emph{exponential} improvement in running time over~\cite{BP19}, the only previous algorithm to give optimal-size fault tolerant spanners.  And unlike~\cite{DR20,DinitzK:11} it gives spanners with existentially optimal size, saving an $O(k)$ factor over~\cite{DR20} and an $O(f\log n)$ factor over~\cite{DinitzK:11}.  It is also polynomially faster than both~\cite{DR20,DinitzK:11}.  We note that not only is $k = \Theta(\log n)$ (and superconstant $k$ more generally) a particularly interesting regime (as discussed), large values of $f$ are also particularly interesting.  If we only ever think of $f$ as small then the dependence on $f$ in the size of the spanner does not matter much, but of course we are interested in protecting against as many faults as possible!  So our results are strongest (compared to previous work) precisely in one of the most interesting regimes for fault-tolerant spanners: $k = \Theta(\log n)$ and $f$ polynomially large in $n$.

\paragraph{Additional Properties.} Our algorithms and techniques have a few other properties that we briefly mention here, but which will not be a focus in the paper.
First, a corollary of our techniques and analysis is that we actually speed up the running time of the \emph{non-fault tolerant} greedy algorithm from $O(mn^{1+1/k})$ to $O(k n^{2+1/k})$.  Second, our algorithms and bounds continue to hold for edge fault tolerance, but for simplicity we will only discuss the VFT case.  These size bounds are also not known to be optimal for edge fault tolerance (except for spanners of stretch $3$) since the known lower bounds are weaker, making our results more interesting in the VFT setting.  See Section~\ref{sec:conclusion} for more discussion of edge fault tolerance.
Finally, since our algorithms are slack-greedy, they are \emph{unconditionally optimal}: even if the Erd\H{o}s girth conjecture is false, our algorithms still produce spanners of optimal size (whatever that is).

\subsubsection{Our First Algorithm \label{sec:firstsimple}}

Our first algorithm is a surprisingly simple randomized algorithm that, while not as efficient as the algorithm we will use to prove Theorem~\ref{thm:randomized-main}, achieves our main goal: it has polynomial running time and produces spanners of optimal size.  It also illustrates the main ideas that our more advanced algorithms (faster and/or deterministic) will utilize.
\begin{theorem} \label{thm:basic-main}
There is a randomized algorithm which, given an undirected weighted $n$-node graph and positive integers $f$ and $k$, runs in polynomial time
and with high probability returns an $f$-VFT $(2k-1)$-spanner with $O(f^{1-1/k} n^{1+1/k})$ edges.
\end{theorem}

The main new ingredient is the following simple $\test(u, v)$ subroutine.  To test an edge $(u, v)$ we first randomly sample $\Theta(\log n)$ induced subgraphs of the current spanner, each of which is obtained by including $u$ and $v$ and then including each other node with probability $1/(2f)$.
Then we test \eqref{eq:nonfttest} in each of the sampled subgraphs, and we add $(u, v)$ to the spanner iff a large enough fraction of these subgraphs violate \eqref{eq:nonfttest}.
For correctness, one observes that if there exists any set of $|F|=f$ vertex deletions under which $\dist_{H \setminus F}(u, v)$ is large, then with high probability a large fraction of the subgraphs will delete all of $F$.
Hence $\dist(u, v)$ will be large in these subgraphs, and we will correctly include $(u, v)$ in the spanner.

The more interesting part of the argument is bounding the size of the output spanner.
This relates to the \emph{blocking set technique} introduced in \cite{BP19} and also used in~\cite{DR20}.  This technique uses the following observation: if the spanner $H$ is sufficiently dense, and one samples a random induced subgraph on $n/f$ nodes, then that subgraph will still be dense enough that it must have some $\le 2k$-cycles.
This statement is even somewhat robust, in that one will not be able to destroy all $\le 2k$-cycles by deleting only a constant fraction of remaining edges in the subgraph.
Thus one can certify \emph{sparsity} of the spanner by arguing that it is in fact possible to sample a subgraph, remove a constant fraction of the surviving edges, and destroy all $\le 2k$-cycles in the process.
In \cite{BP19, DR20}, one roughly predefines a small set of edges in the spanner (called the ``blocking set'') that intersects all $\le 2k$-cycles; the strategy is then to sample a subgraph and remove the parts of the blocking set that survive.
This approach works, but it is limited: if we use a slack FT-greedy algorithm (as in~\cite{DR20}) then the size of the blocking set increases with the slack, giving spanners that no longer have optimal size. 

Our main idea is to bypass blocking sets by more closely tying together the algorithm and the analysis.  The analysis of the blocking set technique uses random subgraph sampling, but this does not appear in the actual \emph{algorithm} of~\cite{DR20} or~\cite{BP19}.  Our new algorithm, by explicitly sampling subgraphs as part of the $\test(u,v)$ subroutine, is in some sense doing algorithmically exactly the \emph{minimum} needed for the analysis to work.  As is shown more formally in Section~\ref{sec:basic}, if an edge passes our new $\test(u,v)$ then \emph{by construction} the probability that it will have to be removed in the analysis in order to obtain high-girth is at most some constant less than $1$.  So getting our sampled subgraph to be high-girth in the analysis requires removing only a constant fraction of the remaining edges.  This is in contrast to the ideas behind blocking sets, where the blocking set is predefined by the algorithm and so an edge either has to be removed from the analysis subgraph (if it is part of the blocking set) or does not (if it isn't).  Hence another way to think of our idea is that we are moving from a \emph{global} analysis of which edges need to be removed (blocking sets, where edges are either in the set or not) to a \emph{local} analysis (where each edge has only a constant probability of being removed).

Details of this algorithm and analysis are given in Section \ref{sec:basic}.

\subsubsection{Our Faster Randomized Algorithm \label{sec:randomized-overview}}

As discussed, obtaining \emph{fast} algorithms for spanners -- not just any polynomial time -- is a long-standing and important line of research.
Once we are able to achieve polynomial time, we naturally want to minimize this time.
Our next algorithm (which achieves the running time bound of Theorem~\ref{thm:randomized-main}) optimizes the runtime in two ways:

\begin{enumerate}
    \item It is costly to randomly sample subgraphs anew in each round of the algorithm.
    A more efficient approach is to randomly sample vertex sets once at the very beginning, use these to perform the $\test$ in each round, and incrementally maintain the sampled subgraphs as edges are added to the spanner.
    We show that this approach works as long as we sample $\Theta(f^3 \log n)$ total subgraphs in the beginning.
    
    \item Since we incrementally maintain a fixed set of subgraphs, to test (\ref{eq:nonfttest}) on each subgraph we can use an appropriate incremental data structure rather than computing from scratch each time.
    The obvious way of doing this requires using an incremental dynamic distance oracle / APSP algorithm, but unfortunately known constructions are not fast enough for our purposes.  However, we can use an idea from \cite{RZ11}: it suffices to solve a certain relaxed version of this problem.
    Specifically, instead of measuring $\dist(u, v)$ exactly; it suffices for our oracle to simply decide whether $\dist(u, v) > (2k-1)w(u, v)$ in each subgraph.
    In the setting of unweighted input graphs, this easily reduces to a problem of \emph{reachability} (rather than distance) on a $(k+1)$-layered version of the subgraph, and we can use an observation from~\cite{DR20} to make this reduction work for weighted input graphs as well thanks to the fact that our framework is a slack version of the greedy algorithm.
    We can then use a classical data structure for incremental reachability by Italiano \cite{Italiano86}.
\end{enumerate}

Together, these improvements give the runtime listed in Theorem \ref{thm:randomized-main}.
Details are given in Section \ref{sec:randomized}.

\subsubsection{Our Deterministic Algorithm}

Our second improvement is to regain determinism.  Both the exponential time algorithm of~\cite{BDPW18} and the polynomial time algorithm of~\cite{DR20} are deterministic, while the core ideas of our previous two algorithms seem to require randomization (particularly our first, non-optimized algorithm).  But by constructing set systems with specific properties through the use of almost-universal hash functions, we are able to derandomize the algorithm of Theorem~\ref{thm:randomized-main}.

The main idea is to leverage the fact that our fast randomized algorithm samples vertex sets only once at the very beginning.  By examining the proof of Theorem~\ref{thm:randomized-main} we can determine what properties we need these sets to have.  Informally, we need that for every $(u,v) \in E$ there are not many sets containing both $u$ and $v$, and that for every $(u,v) \in E$ and $F \subseteq V \setminus \{u,v\}$ with $|F| \leq f$, a constant fraction of the sets which contain both $u$ and $v$ do not contain any vertex in $F$ (note that this guarantee has to hold simultaneously for all possible fault sets).  So we just need to give a deterministic construction of such a set system.  We show how to do this by building an (almost-)universal hash family from $V$ to $[\Theta(f)]$, and for each hash function in the family creating $\Theta(f^2)$ sets based on pairs of hash values.  Interestingly, we are able to tolerate relatively large values of ``almost'': our spanner construction still works even if the universal hashing guarantee is violated by large constants.  Most of the literature on hashing, on the other hand, is optimized for the case of only a $(1+\epsilon)$ violation.  By taking advantage of our ability to withstand weaker hashing guarantees, we can design an extremely small hash family based on \emph{Message Authentication Codes} from cryptography, which is a standard and classical construction but to the best of our knowledge has not previously been used in the context of hashing or derandomization.  This allows us to obtain running time that is essentially the same as our fast randomized algorithm when $f$ is polynomial in $n$.  

These ideas give the deterministic runtime in Theorem \ref{thm:deterministic-main}.
Details are given in Section \ref{sec:deterministic}.  Independently, Karthik and Parter~\cite{karthik2021deterministic} used similar ideas but in a more sophisticated manner to provide improved derandomizations for a number of related combinatorial objects, and their techniques when applied to our algorithm make it possible to remove the restriction that $f$ is at least polynomial in $n$.

\section{Preliminaries and Notation} \label{sec:notation}

We will use $\tilde O(\cdot)$ to suppress polylogarithmic (in $n$) factors.  For any integer $a \geq 1$, let $[a] = \{1,2,\dots, a\}$.  Given an edge-weighted graph $G = (V, E, w)$, let $d_G(u,v)$ denote the shortest-path distance from $u$ to $v$ in $G$ according to the weight function $w$ and let $d^*_G(u,v)$ denote the unweighted distance (minimum number of hops) from $u$ to $v$ in $G$.

Many of our algorithms are randomized, and so we make claims that hold \emph{with high probability}.  Formally, this means that they hold with probability at least $1 - 1/n$.\footnote{By changing the constants in the algorithm/analysis, all high probability claims we make can be made to hold with probability at least $1-1/n^c$ for any constant $c$.  We choose $c=1$ only for simplicity.}

We will use the following Chernoff bounds (see~\cite{DP09}).
\begin{theorem} \label{thm:Chernoff}
    Let $X = \sum_{i=1}^{n} X_i$, where $X_i$ ($i \in [n]$) are independently distributed in $[0,1]$.  Then:
    \begin{itemize}
        \item For $0 < \epsilon < 1$:  $\Pr[X < (1-\epsilon) \E[X]] \leq \exp((-\epsilon^2 / 2) \E[X])$ and $\Pr[X > (1+\epsilon) \E[X]] \leq \exp((-\epsilon^2 / 3) \E[X])$.
        \item If $t > 2 e \E[X]$, then $\Pr[X > t] < 2^{-t}$.
    \end{itemize}
\end{theorem}

We will use the following structural lemma about fault-tolerant spanners, which was given explicitly in \cite{DR20} but appeared implicitly in essentially all previous papers on fault-tolerant spanners.  It essentially says that we only have to worry about spanning edges (not all pairs of nodes), and only edges for which the shortest path between the endpoints is the edge itself.

\begin{lemma} \label{lem:structure}
Let $G = (V, E)$ be a graph with weight function $w$ and let $H$ be a subgraph of $G$.  Then $H$ is an $f$-VFT $t$-spanner of $G$ if and only if $d_{H \setminus F}(u,v) \leq t \cdot w(u,v)$ for all $F \subseteq V$ with $|F| \leq f$ and $u,v \in V \setminus F$ such that $(u,v) \in E$.
\end{lemma}

\section{Optimal Fault-Tolerant Spanners in Polynomial Time} \label{sec:basic}

In this section we resolve the main open question left by~\cite{DR20,BP19, BDPW18} by proving Theorem~\ref{thm:basic-main}: we give a polynomial time algorithm which constructs optimal-size vertex fault tolerant spanners.  As discussed in Section~\ref{sec:firstsimple}, the algorithm itself is quite simple: we just use the greedy algorithm but test whether to include an edge by sampling subgraphs and checking whether the distance between the two endpoints is too large.
This algorithm is given formally as Algorithm~\ref{ALG:basic}.

\begin{algorithm}
\caption{Basic $f$-VFT $(2k-1)$-Spanner Algorithm}
\label{ALG:basic}
\begin{algorithmic}[1]
\REQUIRE Graph $G = (V, E)$ on $n$ nodes, edge weights $w : E \rightarrow \mathbb{R}^+$, integers $k \geq 1$ and $f \geq 1$.
\STATE $H\leftarrow (V, \emptyset)$
\FORALL{$e = (u,v)\in E$ in nondecreasing weight order}
    \STATE{Sample $\alpha = c \log n$ subgraphs $\{\widehat{H^i_e} \subseteq H\}_{i \in [\alpha]}$, where each is an induced subgraph on a vertex set obtained by including $u$ and $v$, and then each other node independently with probability $1/(2f)$.}
    \STATE{Let $\widehat{P_e}$ be the fraction of these subgraphs in which $d_{\widehat{H^i_e}}(u, v) > (2k-1) \cdot w(u,v)$.}
    \IF{$\widehat{P_e} \ge 1/4$} \STATE Add $e$ to $H$ \ENDIF 

\ENDFOR
\RETURN H
\end{algorithmic}
\end{algorithm}

The following definitions will be useful in our analysis.
Let $H$ be the final spanner, and let $H'$ be an induced subgraph of $H$ obtained by including each node, independently, with probability $1/(2f)$ (note that $H'$ is only an analytical tool, not part of the algorithm).
Let $H_e$ ($H'_e$) denote the subgraph of $H$ ($H'$) containing only the edges considered strictly before $e$ in the algorithm.
Let
$$P_{e=(u, v)} := \Pr\left[ d_{H'_e}(u, v) > (2k-1)\cdot w(u, v) \mid u, v \in V(H') \right]$$
where the probability is over the random construction of $H'$.
So $\widehat{P_e}$ in the algorithm is an experimental estimate of $P_e$, and we can bound its accuracy as follows:
\begin{lemma} \label{lem:accuracy}
With high probability, for every edge $e \in E$, we have $\widehat{P_e} \in P_e \pm 1/8$. 
\end{lemma}
\begin{proof}
We will prove the lower bound $\widehat{P_e} \ge P_e - 1/8$; the upper bound is essentially identical.
The random variable $\alpha \widehat{P_e}$ is the sum of $\alpha$ random variables $\left\{\widehat{P_e^i} \right\}$, where
\[
\widehat{P_e^i} := \begin{cases}
0 & \text{if } d_{\widehat{H^i_e}}(u, v) \le (2k-1) \cdot w(u, v) \\
1 & \text{if }d_{\widehat{H^i_e}}(u, v) > (2k-1) \cdot w(u, v).
\end{cases}
\]
Thus we may apply Chernoff bounds (Theorem~\ref{thm:Chernoff}), giving:
\begin{align*}
   \Pr\left[\alpha \widehat{P_e} < \alpha (P_e - 1/8)\right] &= \Pr\left[\alpha \widehat{P_e} < \E\left[\alpha \left(\widehat{P_e} - 1/8 \right)\right]\right] \le \Pr\left[\alpha \widehat{P_e} < \frac78 \E\left[\alpha \widehat{P_e}\right]\right]\\
    &\leq e^{-\frac{(1/8)^2}{2} \alpha \E\left[\widehat{P_e}\right]} \leq e^{-\alpha / 128}
\end{align*}
where the first equality follows by linearity of expectation, and the fact that by construction $P_e = \E\left[\widehat{P_e}\right]$.
If we set $\alpha \ge 128 \cdot 3 \ln n$, then the probability is at most $1/n^3$.
So, by a union bound over the $m \le n^2$ edges in the input graph, the probability that the lower bound fails for \emph{any} edge is at most $1/n$, proving the lemma.
\end{proof}

We are now ready to prove the properties of Algorithm \ref{ALG:basic}.

\begin{lemma} \label{lem:basic-correct}
With high probability, Algorithm~\ref{ALG:basic} returns an $f$-VFT $(2k-1)$-spanner.
\end{lemma}
\begin{proof}
Let $e = (u, v)$ be an edge considered by the algorithm, and suppose there exists a fault set
$F \subseteq V \setminus \{u,v\}$ with $|F| \leq f$ such that
$$d_{H_e \setminus F}(u,v) > (2k-1) \cdot w(e).$$
In the event that $H'_e$ contains $u, v$ but it does not contain any node in $F$, we thus also have
$$d_{H'_e}(u, v) > (2k-1) \cdot w(e).$$
Thus $P_e$ is \emph{at least} the probability that none of the nodes in $F$ survive in $H'_e$, which we may bound:
\[
P_e \ge (1-p)^{|F|} \geq \left(1- \frac{1}{2f}\right)^f  \geq 1/2.
\]
By Lemma \ref{lem:accuracy}, with high probability we have $\widehat{P_e} \ge P_e - 1/8 \ge 1/2 - 1/8 > 1/4$, and so we add $e$ to $H$ in the algorithm.
So for any edge $e$ \emph{not} added to the spanner, no such fault set $F$ exists.
It then follows from Lemma \ref{lem:structure} that $H$ is an $f$-VFT $(2k-1)$-spanner.
\end{proof}

\begin{lemma} \label{lem:basic-size}
    With high probability, $|E(H)| \leq O\left(f^{1-1/k} n^{1+1/k} \right)$. 
\end{lemma}
\begin{proof}
Recall that $H'$ is an induced subgraph of $H$ obtained by including every vertex independently with probability $1/(2f)$, and hence $\E[|V(H')|] = n/(2f)$.  For each edge $(u, v) \in E(H')$, we say that
$(u, v)$ is \emph{bad} if $\dist_{H'_e}(u, v) \le (2k-1) \cdot w(u, v)$, and otherwise (if $\dist_{H'_e}(u, v) > (2k-1) \cdot w(u, v)$) we say that $(u, v)$ is \emph{good}.
Let $H'' \subseteq H'$ be obtained by deleting all bad edges.
We will now bound its expected number of edges $\E\left[|E(H'')|\right]$, conditioned on the high probability event from Lemma \ref{lem:accuracy} holding, in two different ways:
\begin{itemize}
    \item For any cycle $C$ in $H'$ with at most $2k$ edges, notice that the edge $(u, v) \in C$ considered last by the algorithm is bad, since there is a $u \leadsto v$ path around the cycle consisting of at most $2k-1$ edges, each of weight at most $w(u, v)$.
    Thus $(u, v)$ is removed in $H''$.
    It follows that $H''$ has \emph{no} cycles with at most $2k$ edges; the folklore Moore Bounds then imply that
    $$\left|E(H'')\right| = O\left(\left|V(H'')\right|^{1+1/k}\right).$$
    Since each of the $n$ nodes in $H$ are included in $H''$ independently with probability $1/(2f)$, we have
    \begin{equation} \label{eq:basic-upper}
    \E\left[\left| E(H'') \right|\right] = O\left(\E\left[ \left|V(H'')\right|^{1+1/k}\right]\right) = O\left( \E\left[ \left|V(H'')\right|\right]^{1+1/k} \right) = O\left(\left(\frac{n}{f}\right)^{1+1/k}\right).
\end{equation}
    In particular, the reason we can pull the exponent outside the expectation in the second step is due to the following computation:
    \begin{align*}
        \left| \E\left[\left|V(H'')\right|^{1+1/k}\right] - \E\left[\left|V(H'')\right|\right]^{1+1/k}\right| &\le \left|\E\left[\left|V(H'')\right|^{2}\right] - \E\left[\left|V(H'')\right|\right]^{2}\right|\\
        &= \mathop{Var}\left[\left|V(H'')\right|\right] = O(n/f)
    \end{align*}
    and hence the difference between these two terms may be hidden in the $O$.

    \item For an edge $e=(u, v) \in E(H)$, the probability that $e$ survives in $H''$ may be decomposed as the probability that it survives in $H'$, times the probability that it survives in $H''$ given that it survives in $H'$.
    This gives:
    \begin{align*}
        \Pr\left[(u, v) \in E(H'') \right] &= \Pr\left[u, v \in V(H')\right] \cdot \Pr\left[(u, v) \text{ is good} \mid u, v \in E(H')\right]\\
        &= \Theta(1/f^2) \cdot P_e.
    \end{align*}
    Since $(u, v)$ was added to $H$, we have $\widehat{P_e} \ge 1/4$, and so (since we condition on Lemma \ref{lem:accuracy}) we have $P_e \ge 1/8$.
    Hence the probability of $(u, v) \in E(H'')$ is $\Theta(1/f^2)$.
    By linearity of expectations, we then have
    \begin{equation} \label{eq:basic-lower}
    |E(H'')| = \Theta\left( \frac{|E(H)|}{f^2} \right).
    \end{equation}
\end{itemize}

Combining \eqref{eq:basic-upper} and \eqref{eq:basic-lower}, we have
\[\Omega\left(\frac{|E(H)|}{f^2}\right) = \E\left[\left|E(H'')\right|\right] = O\left(\left(\frac{n}{f}\right)^{1+1/k}\right)\]
and so, comparing the left- and right-hand sides and rearranging, we get
\begin{align*}
|E(H)| = O\left(n^{1+1/k} f^{1-1/k}\right). \tag*{\qedhere}
\end{align*}
\end{proof}

\begin{lemma} \label{lem:basic-time}
Algorithm~\ref{ALG:basic} runs in polynomial time.
\end{lemma}
\begin{proof}
We first need to sort the edges by weight, which takes at most $O(m \log n)$ time.
Then for each edge $e \in E$, we must sample $\Theta(\log n)$ independent subgraphs and then run a single-source shortest path computation in each.
Sampling such a subgraph takes $O(n + |E(H)|) \leq O(f^{1-1/k} n^{1+1/k})$ time (using the bound on $|E(H)|$ from Lemma~\ref{lem:basic-size}).
The running time of the shortest path computation on the subgraph is at most $O(|E(H)| \log n)$,\footnote{This estimate is conservative; except in a small range of parameters the subgraph is much smaller than $H$ and thus the running time of the shortest path computation is dominated by the time needed to sample the subgraph in the first place.}
which is at most $O(f^{1-1/k} n^{1+1/k} \log n)$.
Since we repeat $O(\log n)$ times per round, and we have $m$ total rounds in the algorithm, the total runtime is $O(m f^{1-1/k} n^{1+1/k} \log^2 n)$.
\end{proof}

\section{An Even Faster Randomized Algorithm \label{sec:randomized}}

Algorithm~\ref{ALG:randomized} is a bit more complicated than Algorithm~\ref{ALG:basic}, but it is significantly faster: in terms of runtime, it essentially turns the $O(m)$ factor in Lemma~\ref{lem:basic-time} into an $O(k^2 n)$ factor.
In Algorithm~\ref{ALG:basic}, up to a logarithmic factor the runtime per sampled subgraph is dominated by the time required to compute the subgraph rather than the time to measure distances on the subgraph, since sampling the subgraph requires time linear in $|E(H)|$ while distances are computed in the subgraph itself, which has fewer edges than $H$.  
This leads to essentially $O(|E(H)| \cdot m \log n)$ total runtime, since we sample $O(m \log n)$ subgraphs in total: $O(\log n)$ subgraphs in each round, and $m$ total rounds.
In our new algorithm, we improve this by (1) pre-sampling $\Theta(f^3 \log n)$ subgraphs and using them in each round, and (2) measuring relevant distances on these subgraphs using a certain incremental dynamic algorithm, rather than recomputing from scratch.
See Section \ref{sec:randomized-overview} for a more detailed overview.
We now state the algorithm formally.

\begin{definition}
Let $G = (V, E)$ be an undirected graph, and let $k \geq 1$ be an integer.  The \emph{layered graph} $G^{2k}$ is the directed graph with vertex set $V \times [2k]$ and edges
\begin{align*}
((u,i), (v, i+1)) & \text{ for all } i \in [2k-1], (u,v) \in E, \text{ and}\\
((u,i), (u, i+1)) & \text{ for all } i \in [2k-1], u \in V.
\end{align*}
\end{definition}

Notice that the \emph{unweighted} distance between $u$ and $v$ in $G$ is at most $2k-1$ if and only if there is a path from $(u,1)$ to $(v, 2k)$ in $G^{2k}$.

\begin{theorem}[\cite{Italiano86}] \label{thm:dynamic-connectivity}
There is a data structure 
which takes $O(n^2)$ time to initialize on an empty $n$-node graph, and which can then support directed edge insertions in $O(n)$ time (amortized) and reachability queries (answering ``is there currently a $u \leadsto v$ path in the graph?) in $O(1)$ time.
\end{theorem}

\begin{algorithm}
\caption{Faster $f$-VFT $(2k-1)$-Spanner Algorithm}
\label{ALG:randomized}
\begin{algorithmic}[1]
\REQUIRE Graph $G = (V, E)$, edge weights $w : E \rightarrow \mathbb{R}^+$, integers $k \geq 1$ and $f \geq 1$.\\

\COMMENT {\textbf{preprocessing phase}}
\STATE $H\leftarrow (V, \emptyset)$

\FOR{$i=1$ \TO $\alpha = c f^3 \log n$}
    \STATE Create $V_i$ by including every vertex of $V$ independently with probability $1/(2f)$
    \STATE Let $H_i = (V_i, \emptyset)$
    \STATE Create the layered graph $H_i^{2k}$, and initialize the data structure of Theorem~\ref{thm:dynamic-connectivity} for $H_i^{2k}$
\ENDFOR
\STATE \textbf{for all} $e = (u,v) \in E$ \textbf{do} Let $L_e = \{i \ \mid \ u,v \in V_i\}$\\~\\

\COMMENT {\textbf{main greedy algorithm}}
\FORALL{$e = (u,v) \in E$ in nondecreasing weight order} 
    \FORALL{$H_i \in L_e$}
        \STATE Query whether there is a $(u,1) \leadsto (v, 2k)$ path in $H_i^{2k}$
    \ENDFOR
    \STATE Let $P_e$ be the fraction of subgraphs $H_i \in L_e$ where the query returns NO
    \STATE Let $0 < \tau < 1$ be an absolute constant that we choose in the analysis
    \IF{$P_e \geq \tau$} 
        \STATE Add $e$ to $H$ 
        \FORALL{$i \in L_e$}
            \STATE Add $e$ to $H_i$  
            \STATE Update the connectivity data structure for $H_i^{2k}$ by inserting the edges $((u,j), (v, j+1))$ and $((v,j), (u, j+1))$ in $H_i^{2k}$ for all $j \in [2k-1]$  
        \ENDFOR
    \ENDIF 
\ENDFOR
\RETURN H
\end{algorithmic}
\end{algorithm}

Like for Algorithm~\ref{ALG:basic}, we let $H$ be the final spanner, and now we let $H'$ be a uniform random subgraph \emph{among those selected in the preprocessing phase}.
Let $H_e$ ($H'_e$) denote the subgraph of $H$ ($H'$) containing only the edges considered strictly before $e$ in the algorithm.
We note that we do not have separate analogous definitions of $\widehat{P_e}, P_e$ this time: the relevant probability $P_e$ is computed \emph{exactly} by the algorithm.
We start our analysis with the following technical lemma:

\begin{lemma} \label{lem:FS-random}
With high probability over the choice of random subgraphs in the preprocessing phase, for every $e=(u, v) \in E$ and $F \subseteq V \setminus \{u, v\}$ with $|F| \le f$, we have:
\begin{enumerate}
     \item $|L_e| = O(f \log n)$
     \item $|\{i \in L_e \ \mid \ F \cap V_i = \emptyset\}| = \Omega(f \log n)$
\end{enumerate}
\end{lemma}
\begin{proof}
For the first part, by linearity of expectations the expected number of sets $V_i$ that contain both $u$ and $v$ is exactly $\alpha / (2f)^2 = (c/4) f \log n$.
Applying Chernoff bounds (Theorem \ref{thm:Chernoff}), we have (for sufficiently large $c$):
\begin{align*}
    \Pr\left[ |L_e| >  ce \cdot f \log n \right] &= \Pr\left[ |L_e| > 4e \E\left[\left| L_e \right| \right] \right] < 2^{-ce \cdot f \log n} < 1/(2 n^3)
\end{align*}
A union bound over the $m \le n^2$ edges in the graph implies that $|L_e| = O(f \log n)$ for all $e$, simultaneously, with probability at least $1 - 1/(2n)$. 

For the second part of the lemma, for any $e=(u, v) \in E$ and $F \subseteq V \setminus \{u, v\}$ with $|F| \leq f$, and for any $i \in [\alpha]$, we have
\begin{align*}
\Pr[i \in L_e \text{ and } F \cap V_i = \emptyset] &= \Pr[u, v \in V_i] \cdot \Pr[ F \cap V_i \ne \emptyset]
\end{align*}
since the two probabilities on the right-hand side consider independent events (since $u, v \notin F$).
Thus we may continue
\begin{align*}
\Pr[i \in L_e \text{ and } F \cap V_i = \emptyset] &= \left(\frac{1}{4f^2}\right) \left(1-\frac{1}{2f}\right)^{|F|}\\
&\geq \frac{1}{4f^2} \left(1-\frac{1}{2f}\right)^f\\
&\geq \frac{1}{8f^2}.
\end{align*}
By linearity of expectations,
$$\E\left[\left|\{i \in L_e \ \mid \ F \cap V_i  = \emptyset\}\right|\right] \geq \frac{\alpha}{8f^2} =\left(\frac{c}{8}\right) f \log n.$$
Again by Chernoff bounds, we have
\begin{align*}
    \Pr\left[ \left|\{i \in L_e \ \mid \ F \cap V_i  = \emptyset\}\right| < \left(\frac{1}{2}\right) \left(\frac{c}{8} \right) f \log n  \right] &\le \exp\left(\left(-\frac{1}{8} \right) \left(\frac{c}{8} \right) f \log n  \right)\\
    &\le \frac{1}{2n^{4f}} \tag*{assuming sufficiently large $c$.}
\end{align*}
Taking a union bound over all $\le n^f$ possible choices of $F$ and all $m \le n^2$ edges, we have $$\left|\{i \in L_e \ \mid \ F \cap V_i  = \emptyset\}\right| = \Omega(f \log n)$$
for all choices of $e, F$, simultaneously, with probability $\ge 1 - 1/(2n)$.
Hence, by an intersection bound, the two parts of the lemma hold jointly with high probability.
\end{proof}

We are now ready to prove the properties of Algorithm~\ref{ALG:randomized}:
\begin{lemma} \label{lem:random-correct}
With high probability, Algorithm~\ref{ALG:randomized} returns an $f$-VFT $(2k-1)$-spanner. 
\end{lemma}
\begin{proof}
Let $e = (u, v)$.  As with our proof of correctness of Algorithm~\ref{ALG:basic} (Lemma~\ref{lem:basic-correct}), by Lemma~\ref{lem:structure} we just need to show that when the algorithm considers $e$, if there is a fault set $F$ for which
$$d_{H_e \setminus F}(u,v) > (2k-1) \cdot w(e),$$
then the algorithm adds $e$ to $H$.
Notice that this implies
$$d^*_{H_e \setminus F}(u, v) > 2k-1,$$
since by construction the weight of every edge in $H_e$ is no larger than $w(e)$ (recall from Section~\ref{sec:notation} that $d^*$ denotes the unweighted distance).
By Lemma~\ref{lem:FS-random}, with high probability a constant fraction of the $i \in L_e$ have $F \cap V_i = \emptyset$ and thus we have $d^*(u, v) > 2k-1$ in the corresponding subgraphs.
Thus $P_e$ is at least an absolute constant; by setting $\tau$ less than this constant, we will add $e$ to $H$.
\end{proof}

\begin{lemma} \label{lem:random-size}
With high probability, $|E(H)| \leq O\left(f^{1-1/k} n^{1+1/k}\right)$.
\end{lemma}
\begin{proof}
This loosely follows the proof of Lemma \ref{lem:basic-size}.
For each edge $(u, v) \in E(H')$, let us say:
\begin{align*}
(u, v) \text{ is \emph{bad}} & \text{ if } \dist^*_{H'_e}(u, v) \le 2k-1\\
(u, v) \text{ is \emph{good}} & \text{ if } \dist^*_{H'_e}(u, v) > 2k-1.
\end{align*}
Let $H'' \subseteq H'$ be obtained by deleting all bad edges.
We now again bound $\E[|E(H'')|]$ in two ways:
\begin{itemize}

    \item By essentially the same argument as in Lemma \ref{lem:basic-size}, $H''$ has no cycles on $\le 2k$ edges, and thus 
    \begin{align*}
        \E\left[ \left| E(H'') \right|\right] &= O\left( \E\left[\left| V(H'') \right|^{1+1/k} \right] \right) = O\left(\E\left[ \left| V(H'') \right| \right]^{1+1/k} \right)
    \end{align*}
    where we may pull the exponent outside the expectation by the same argument as in Lemma \ref{lem:basic-size}.
    When we choose random subgraphs in the preprocessing phase, the \emph{total} number of nodes added to all subgraphs can be viewed as the sum of independent binary random variables, and the expectation is $\Theta(n f^2 \log n)$.
    Thus, by Chernoff bounds, with high probability over the choice of random subgraphs in the preprocessing phase, we do indeed have $\Theta(n f^2 \log n)$ total nodes in our subgraphs.
    Conditioned on this high probability event, since we have exactly $\alpha = \Theta(f^3 \log n)$ subgraphs, the average sampled subgraph $H_i$ has $|V(H_i)| = \Theta(n / f)$.
    We then have $\E[|V(H'')|] = \Theta(n/f)$, and so with high probability
    $$\E\left[\left| E(H'') \right|\right] = O\left( \left( \frac{n}{f} \right)^{1+1/k} \right)$$
    (where the expectation is only over the choice of $H'$ among the subgraphs sampled in the preprocessing phase).

    \item By Lemma \ref{lem:FS-random}, with high probability over the choice of random subgraphs, we have $|L_e| = \Theta(f \log n)$ for all edges $e$.
    In this event, whenever we add an edge $e=(u, v)$ to the spanner we also add the edge to $\Theta(f \log n)$ out of the $\alpha = \Theta(f^3 \log n)$ subgraphs.
    Moreover, by construction, in at least a constant $\tau$ fraction of these subgraphs $H_i$, our $(u, 1) \leadsto (v, 2k)$ path query returns NO.
    It follows that the current unweighted $u \leadsto v$ distance in these subgraphs is $>2k-1$, and hence $e$ is a \emph{good} edge in $\Theta(\tau \cdot f \log n) = \Theta(f \log n)$ subgraphs.
    So the total number of good edges among the $\Theta(f^3 \log n)$ subgraphs is $\Theta(|E(H)| f \log n)$
    We thus have
    $$\E\left[|E(H'')| \right] = \Theta\left(\frac{|E(H)|}{f^2}\right).$$
\end{itemize}

Combining these, we have
$$\Omega\left( \frac{|E(H)|}{f^2} \right) = \E\left[|E(H'')|\right] = O\left( \left( \frac{n}{f} \right)^{1+1/k} \right),$$
and again the lemma follows by comparing the left- and right-hand sides and rearranging.
\end{proof}

\begin{lemma} \label{lem:randomized-running-time}
The expected running time of Algorithm~\ref{ALG:randomized} is at most \[O\left(k^2 f^{1-1/k} n^{2 + 1/k} \log n  + m f^2 \log n\right).\]
\end{lemma}
\begin{proof}
We first analyze the preprocessing phase.  Note that without loss of generality, $m \geq f^{1-1/k} n^{1+1/k} \geq fn$ or else we are already finished (we can simply return the input graph).

We can create the $V_i$ sets in time $O(n f^3 \log n) = O(mf^2 \log n)$ by flipping $\alpha = O(f^3 \log n)$ weighted coins for each vertex.  Once we do this, we may assume that every vertex $v$ has a \emph{sorted} list $L_v$ of the values of $i$ for which $v \in V_i$.  In expectation, each of the layered graphs $H_i^{2k}$ has $O(kn/f)$ vertices, so we can then create the layered graphs in expected time $O((kn/f) f^3 \log n) = O(k n f^2 \log n) = O(k f^{1-1/k} n^{2+1/k} \log n)$ (by linearity of expectations).  Note that this includes the initial edges of the form $((u,i), (u,i+1))$ in each of the layered graphs.  Initializing the data structure of Theorem~\ref{thm:dynamic-connectivity} for each of the layered graphs takes time $O((kn/f)^2)$ in expectation (since the number of nodes is binomial random variable and $\E[X^2] \leq O(\E[X]^2)$ for any binomial random variable $X$), and thus the total expected time to initialize these data structures is $O((kn/f)^2 f^3 \log n) = O(k^2 f n^2 \log n) \leq O(k^2 f^{1-1/k} n^{2+1/k} \log n)$.  We also need to insert all of the initial edges of the form $((u,i), (u,i+1))$ into these data structures, which takes expected time $O(f^3 \log n \cdot (kn/f) \cdot (kn/f)) = O(k^2 f n^2 \log n) \leq O(k^2 f^{1-1/k} n^{2+1/k} \log n)$ (again using the square of a binomial random variable).  Hence we can construct and initialize the vertex sets and the needed data structures in expected time at most
\begin{align*}
    & O(mf^2 \log n) + O(k f^{1-1/k} n^{2+1/k} \log n) + O(k^2 f^{1-1/k} n^{2+1/k} \log n) + O(k^2 f^{1-1/k} n^{2+1/k} \log n) \\
    =\ & O(mf^2 \log n) + O(k^2 f^{1-1/k} n^{2+1/k} \log n). 
\end{align*}

To create the $L_e$ sets we need to be a little more careful, since doing it naively (looping through each edge and each $i \in [\alpha]$ and checking if both endpoints are in $V_i$) would take $O(m f^3 \log n)$ time.  But we can speed this up: since every vertex $v$ has $L_v$ in a sorted list, for each edge $e = (u,v)$ we can just do a single pass through $L_u$ and $L_v$ to compute $L_e = L_u \cap L_v$.  Thus this can be done in $O(|L_u| + |L_v|)$ time, and since $|L_v| \leq O(f^2 \log n)$ with high probability for every $v \in V$, this takes time $O(f^2 \log n)$ per edge and thus $O(m f^2 \log n )$ total.  

Putting all of this together, we get that the preprocessing takes time $O(k^2 f^{1-1/k} n^{2+1/k} \log n + m f^2 \log n)$.

We now analyze the main loop.  For every $e = (u,v) \in E$, the algorithm performs a connectivity query in $O(f \log n)$ of the layered subgraphs.  By Theorem~\ref{thm:dynamic-connectivity}, this takes $O(m f \log n)$ total time.  When we decide to add an edge $e$ to the spanner (which happens at most $O(f^{1-1/k} n^{1+1/k})$ times by Lemma~\ref{lem:random-size}), we have to do $O(k)$ insertions into each of the $O(f \log n)$ layered graphs in $L_e$.  The amortized cost of each insertion is $O(kn/f)$ by Theorem~\ref{thm:dynamic-connectivity} (since the number of nodes in each layered graph is $O(kn/f)$ by Lemma~\ref{lem:FS-random}), and hence the total time of all insertions is $O(f^{1-1/k} n^{1+1/k} (kn/f) kf \log n) = O(k^2 f^{1-1/k} n^{2+1/k} \log n)$.  This is asymptotically larger than $m \log n$ since $m < n^2$, and hence the running time of the main loop is $O(k^2 f^{1-1/k} n^{2+1/k} \log n)$.
\end{proof}

\section{Deterministic Algorithm \label{sec:deterministic}}

We now design a deterministic algorithm by derandomizing Algorithm~\ref{ALG:randomized}.  Recall that Algorithm~\ref{ALG:randomized} uses randomization in the preprocessing phase to create $\Theta(f^3 \log n)$ vertex sets.
We will derandomize this by \emph{deterministically} creating sets with the same properties by using appropriately chosen hash functions.  This results in a deterministic algorithm whose running time depends on the size of the hash family that we use.  If we use universal or pairwise-independent hash functions, then we end up paying additional factors of $n$ in the running time.  We can improve this by using \emph{almost}-universal hash functions, since our analysis is robust to changes in constants.  Standard constructions then give the same dependence on $n$ as in Algorithm~\ref{ALG:randomized}, but polynomially worse dependence on $f$.  But in the most important regime where $f$ is polynomial in $n$, we can use ideas from \emph{message authentication codes} in cryptography to design a hash family which is significantly more efficient, allowing us to get running time that is essentially \emph{identical} to the randomized algorithm!

\subsection{Set System}

Intuitively, we want sets which ``act like'' the $\Theta(f^3 \log n)$ random sets of Algorithm~\ref{ALG:randomized}.  So they should each have size about $n/f$, there shouldn't be too many sets in each $L_e$ (each edge shouldn't be in too many of the subsets), and for every fault set $F$ a constant fraction of the sets in $L_e$ should not intersect $F$.  We will proceed somewhat similarly to the approach of~\cite{Par19}, who needed a set system where for all sets $A$ of some size $a$ and all sets $B$ of some size $b$, there was at least one set in the system which contained all of $A$ and none of $B$.  This is similar to what we want, but differs in two important respects: we are only concerned with the special case of $a = 2$, but we want not just that there \emph{exists} a set in the system which contains all of $A$ and none of $B$, but that a \emph{constant fraction} of the sets in the system which contain $A$ do not contain any of $B$.

Trying to apply~\cite{Par19} as a black box, or even using the construction from~\cite{Par19}, gives highly suboptimal bounds: the number of sets we would need would be exponential in $f$.\footnote{It is worth noting, though, that it is not hard to change the construction of~\cite{Par19} to give bounds that are polynomial in $f$: one simply needs to modify the construction in their Lemma 17 to define sets based on a hash value \emph{being} some value, as opposed to the given construction which defines sets based on a hash value \emph{not being} some value.} So we need to change the construction. And in order to optimize the running time we will hash onto a smaller range (approximately $f$, whereas using~\cite{Par19} would hash onto approximately $f^2$) and will use a different hash family (at least for the regime where $f$ is polynomial in $n$).

\subsubsection{Almost Universal Hashing}

Parter~\cite{Par19} began with ``almost-pairwise independent" hash families.  We will use a slightly different (but related) definition of ``almost-universal" hash families, which is a weaker requirement but is sufficient for our needs and will allow us to design faster algorithms than if we required almost-pairwise independence.

\begin{definition} \label{def:universal}
A family $\mathcal H = \{ h : U \rightarrow R\}$ is $\delta$-almost universal if:
\begin{enumerate} 
\item For all $x_1, x_2 \in U$ with $x_1 \neq x_2$: $\Pr_{h \sim H}[h(x_1) = h(x_2)] \leq \frac{\delta}{|R|}$, and
\item For all $y \in R$ and $h \in \mathcal H$, $|\{x \in U : h(x) = y\}| \leq O(|U| / |R|)$.
\item Each $h \in \mathcal H$ can be described with $\tilde O(1)$ bits and can be evaluated in $\tilde O(1)$ time.  
\end{enumerate}
\end{definition}

Setting $\delta = 1$ recovers the standard definition of universal hash families.  The second and third parts of this definition are not always part of the standard definition of universality, but easily follow from most standard constructions (and, in particular, from the constructions that we will use).  The third property implies that in time $\tilde O(|\mathcal H| |U|)$ we can compute all hash functions from the family on all elements of the domain.  

We will use two different constructions of $\delta$-almost universal hash families, one which works for all regimes and one which only gives meaningful bounds when the size of the range is polynomial in the size of the domain.  The first construction we will use is the following:
\begin{theorem} \label{thm:AGHP}
For every $\delta > 1$, there is a $\delta$-almost universal hash family $\mathcal H$ with $O(\text{poly}(1/(\delta-1)) |R|^4 \log^2 |U|)$ functions. 
\end{theorem}
To the best of our knowledge, this theorem does not appear explicitly anywhere in the literature, since most papers just bound the number of random bits as $O(\log |R| + \log\log |U|)$.  Since the constant is unspecified, this is not enough to prove Theorem~\ref{thm:AGHP}.  But Theorem~\ref{thm:AGHP} can easily be derived from Theorem 2 of~\cite{AGHP92} by using $N = |U| \log |R|$ bits to define a hash function and setting $k = 2 \log |R|$ and $\epsilon = (\delta-1) / |R|^2$.  

If $|R|$ is close to $|U|$, we can use a different construction based on ideas from the cryptography literature, and in particular from \emph{Message Authentication Codes} (MACs).  This construction and analysis is essentially standard (see~\cite{KR09,Boer93,Tay94,BJKS93}), but to the best of our knowledge has not been explicitly phrased as a hash function before in the literature.  
We will actually use a slightly weaker (and thus more efficient) version of the standard construction since we only need almost-universality, not almost-pairwise independence.  We give the proof for completeness.

\begin{theorem} \label{thm:high-f-universal}
There is a $\left\lceil \frac{\log |U|}{\log |R|}\right\rceil$-almost universal hash family $\mathcal H$ with $O(|R|)$ functions.  
\end{theorem}
\begin{proof}
Without loss of generality, let $U = \{0,1\}^u$ and let $R = \{0,1\}^r$ with $u$ divisible by $r$, and we interpret $U$ as $\mathbb{F}_{2^u}$ and $R$ as $\mathbb{F}_{2^r}$.  A single hash function is defined by a single element $a \in \mathbb{F}_{2^r}$.  Given $x \in U$, we split $x$ into $u/r$ \emph{chunks} $x_0, x_1, x_{(u/r) - 1}$, each of which has $r$ bits and so is an element of $\mathbb{F}_{2^r}$.  This defines a polynomial $M_x(a) = \sum_{i=0}^{(u/r)-1} x_i a^i$ of degree $(u/r)-1$.  Given an element $a \in \mathbb{F}_{2^r}$, we define a hash function $h_{a}(x) = M_x(a)$. 

Let $\mathcal H = \{h_{a} : a \in \mathbb{F}_{2^r}\}$.  Clearly $|\mathcal H| = 2^r = |R|$.  It is also not hard to see that $\mathcal H$ is a $u/r$-almost universal family.  To see this, let $x, y \in U$ with $x \neq y$.  Then $h_a(x) = h_a(y)$ if and only if $M_{x}(a) = M_y(a)$, which is equivalent to $M_x(a) - M_y(a) = 0$.  Clearly $M_x(a) - M_y(a) = \sum_{i=0}^{(u/r)-1} (x_i - y_i) a^i$ is a non-zero polynomial (in $a$) of degree at most $u/r$, so there are at most $u/r$ roots and thus the probability that we choose an $a$ which satisfies this is at most $(u/r) / 2^r = (u/r) / |R|$.

Clearly we can compute these functions quickly enough, so the third property of Definition~\ref{def:universal} holds.  The second property of Definition~\ref{def:universal} also holds, since if we divide the possible $x \in U$ into equivalence classes by everything \emph{except} their lowest-order chunk $x_0$ (so each class has $2^r$ elements and there are $2^{u-r}$ classes), then for every fixed $a$ and every equivalence class there is exactly one element from the class which gets hashed to every possible value.
\end{proof}

Since $\log |U| / \log |R|$ is a constant if $R$ and $U$ are polynomially related, Theorem~\ref{thm:high-f-universal} gives an $O(1)$-almost universal family in the case of $|R| \geq \poly(|U|)$.  We will only use this theorem in that regime.

\subsubsection{Creating Our Sets}
We will use $\delta$-almost universal hash families to create the subsets in the preprocessing stage of Algorithm~\ref{ALG:randomized} rather than creating these sets randomly.  More formally, rather than create $O(f^3 \log n)$ sets independently, we will use the following construction.

Let $\mathcal H$ be a $\delta$-almost universal hash family with domain $V$ and range $[4\delta f]$.\footnote{This might seem strange in conjunction with Theorem~\ref{thm:high-f-universal}, since then the range is a function of $\delta$ but $\delta$ is also a function of the range.  We show how to set the parameters appropriately when we actually instantiate this hash family in the proof of Theorem~\ref{thm:deterministic-main}.}   We will create $|\mathcal H| \binom{4\delta f}{2} = \Theta(|\mathcal H| f^2)$ sets as follows: for every $h \in H$ and $y,z \in [4\delta f]$ with $y \neq z$, we let 
\[
V_{h, \{y,z\}} = \{v \in V: h(v) \in \{y,z\}\}
\]

In order to keep our previous notation, we will let $\alpha = |\mathcal H| \binom{4\delta f}{2}$ and will arbitrarily number these sets and refer to them as $V_1, \dots, V_{\alpha}$.  And as before, we will let $L_e = \{i \in [\alpha] : e \subseteq V_i\}$ be the sets which contain both endpoints of $e$.  We first need the equivalent of Lemma~\ref{lem:FS-random}, but now all bounds are deterministic.

\begin{lemma} \label{lem:FS-deterministic}
If $\mathcal H$ is $\delta$-universal hash family for some constant $\delta > 1$ then there exist constants $c_1, c_2, c_3 > 0$ so that for sufficiently large $n$ the sets $V_1, \dots, V_{\alpha}$ have the following properties:
\begin{enumerate}
    \item $|V_i| \leq c_1 n/f$ for all $i \in [\alpha]$,
    \item $|L_e| \leq c_2 \delta |\mathcal H|$ for all $e \in E$, and
    \item $|\{i \in L_e : F \cap V_i = \emptyset\}| \geq c_3 |\mathcal H|$ for all $e \in E$ and $F \subseteq V$ with $|F| \leq f$ and $F \cap e = \emptyset$.
\end{enumerate}
\end{lemma}
\begin{proof}
The second property of Definition~\ref{def:universal} implies that each set has size at most $O(n/(4 \delta f)) = O(n/f)$, so the first property is clearly true.  For the second property, let $e = (u,v) \in E$, and let $h \in \mathcal H$.  If $h(u) = h(v)$ then $e$ is contained in $4\delta f - 1$ of the $\binom{4\delta f}{2}$ subsets defined by $h$, and if $h(u) \neq h(v)$ then $e$ is contained in $1$ of the subsets defined by $h$.  By the first part of the definition of almost universal, $\Pr_{h \sim \mathcal H}[h(u) = h(v)] \leq \frac{\delta}{4\delta f} = \frac{1}{4f}$.  Thus
\begin{align*}
    |L_e| &\leq \sum_{h \in \mathcal H} 1 + \sum_{h \in \mathcal H : h(u) = h(v)} (4\delta f - 1) \leq |\mathcal H| + \delta |\mathcal H| = (1+\delta) |\mathcal H|.
\end{align*}

For the third property, fix $e = (u,v) \in E$ and  $F \subseteq V \setminus \{u,v\}$ with $|F| \leq f$.  For every $x \in F$, let $X_{xu}$ be the indicator random variable for the event that $h(x) = h(u)$, and similarly let $X_{xv}$ be the indicator random variable for $h(x) = h(v)$.  Note that by the definition of almost-universal, both of these random variables have expectation at most $\delta/(4\delta f) = 1 / (4f)$.  Let $Z \subseteq F$ be the set of elements of $F$ that hash to the same value as $u$ or $v$ (so $Z$ is a random subset).  Then $|Z| \leq \sum_{x \in F} \left(X_{xu} + X_{xv}\right)$, and thus 
\begin{align*}
    \E[|Z|]& \leq \E\left[\sum_{x \in F} (X_{xu} + X_{xv})\right] = \sum_{x \in F} \left(\E[X_{xu}] + \E[X_{xv}]\right) \leq \sum_{x \in F} \frac{2}{4f} \leq \frac12
\end{align*}

So by Markov's inequality, the probability that $|Z| \geq 1$ is at most $1/2$.  Thus in at least half of the hash functions from $\mathcal H$, nothing from $F$ has the same hash value as $u$ or $v$.  If we have such an $h$, \emph{and} if $h(u) \neq h(v)$, then this means that $F \cap V_{h, \{h(u), h(v)\}} = \emptyset$.  As discussed earlier, the number of hash functions in which $h(u) = h(v)$ is at most $|\mathcal H| / (4f)$.  Thus the number of sets which contain both $u$ and $v$ but do not contain any element of $F$ is at least
\begin{align*}
    |\mathcal H| \left(\frac12 - \frac{1}{4f}\right) \cdot 1 \geq \frac38 |\mathcal H|
\end{align*}
as claimed.
\end{proof}

\subsection{The Algorithm and Analysis}

Our deterministic algorithm is the same as Algorithm~\ref{ALG:randomized}, except:
\begin{enumerate}
    \item Instead of sampling $O(f^3 \log n)$ sets, we use the above construction to create $O(|\mathcal H| f^2)$ sets which obey Lemma~\ref{lem:FS-deterministic}, 
    \item We set the threshold $\tau$ to $c_3/c_2$, i.e., we add $e$ to the spanner if $P_e \geq c_3/c_2$ (where $c_2$ and $c_3$ are the constants from Lemma~\ref{lem:FS-deterministic}).
\end{enumerate}

We now prove correctness and size bounds in essentially the same way as in the randomized algorithm.  They just hold deterministically rather than with high probability.

\begin{lemma} \label{lem:deterministic-correct}
This algorithm returns an $f$-VFT $(2k-1)$-spanner. 
\end{lemma}
\begin{proof}
Let $e = (u,v)$.  As with our proofs of correctness of Algorithms~\ref{ALG:basic} and \ref{ALG:randomized}, by Lemma~\ref{lem:structure} we just need to show that if when the algorithm considers $e$ there is a fault set $F$ which is bad for $e$ (i.e., $|F| \leq f$ and $d_{H \setminus F}(u,v) > (2k-1) w(e)$), then the algorithm adds $e$ to $H$.   

Note that for any $i \in [\alpha]$ where $e \subseteq V_i$ but $F \cap V_i = \emptyset$, we have $d_{H_i}(u,v) > (2k-1) w(e)$ (since removing $F$ is enough to make the distance too large) and hence there is no $(2k-1)$-hop path from $u$ to $v$ in $H_i$ (since we consider the edges in nondecreasing weight order).  This implies that our reachability query from $(u,1)$ to $(v, 2k)$ in $H_i^{2k}$ will return NO.  Thus by Lemma~\ref{lem:FS-deterministic}, $P_e \geq \frac{c_3|\mathcal H|}{c_2|\mathcal H|} = \tau$, and so the algorithm will add $e$ to the spanner.  Lemma~\ref{lem:structure} then implies the theorem.
\end{proof}

We can now prove our size bound.

\begin{lemma} \label{lem:deterministic-size}
$|E(H)| \leq O(f^{1-1/k} n^{1+1/k})$.
\end{lemma}
\begin{proof}
We use randomness in our analysis in the same way it was used in the proof of Lemma~\ref{lem:random-correct}, but note that the algorithm itself is deterministic.

Choose an $i \in [\alpha]$ uniformly at random, and let $H' = H_i$.  Then for each $(2k)$-cycle in $H'$, remove the heaviest edge to get $H''$.

By part 1 of Lemma~\ref{lem:FS-deterministic}, we know that $|V(H'')| \leq O(n/f)$.  Since $H''$ has no $(2k)$-cycles, this implies that 
\begin{equation} \label{eq:deterministic-H''-upper}
|E(H'')| \leq O\left((n/f)^{1+1/k}\right).
\end{equation}

Fix some $e = (u,v) \in E(H)$.  Then the third property of Lemma~\ref{lem:FS-deterministic} (with $F = \emptyset$) implies that $\Pr[e \in H'] \geq \Omega(|\mathcal H| / (|\mathcal H| f^2)) =  \Omega(1/f^2)$.  Conditioned on this, $e$ fails to survive to $H''$ only if it was the heaviest edge on some  $(2k)$-cycle in $H'$.  Since the algorithm considers the edges in nondecreasing weight order, $e$ is the heaviest edge on some $(2k)$-cycle if and only if when it was added by the algorithm the hop-distance between $u$ and $v$ in $H'$ was at most $2k-1$. 

The probability of this happening (conditioned on $e$ being in $H'$) is by definition equal to $1-P_e$, and so the probability that $e$ \emph{does} survive to $H''$ (conditioned on being in $H'$) is precisely $P_e$.  Since $e$ was added by the algorithm we know that $P_e \geq c_3/c_2$, and thus the probability that $e$ survives to $H''$ (conditioned on being in $H'$) is at least $c_3/c_2$.  

Hence 
\begin{equation} \label{eq:deterministic-H''-lower}
    \E[|E(H'')|] \geq |E(H)| \cdot \Omega(1/f^2) \cdot (c_3/c_2) \geq \Omega(|E(H)| / f^2)
\end{equation}

Combining \eqref{eq:deterministic-H''-upper} and \eqref{eq:deterministic-H''-lower} implies that
\begin{align*} 
|E(H)| \leq O(f^2 \cdot (n/f)^{1+1/k}) = O(f^{1-1/k} n^{1+1/k})
\end{align*}
as claimed.  
\end{proof}

We now analyze the running time.  Clearly it depends on $|\mathcal H|$, but for now we will leave this as a parameter.  

\begin{lemma} \label{lem:deterministic-running-time}
The running time of the deterministic algorithm is at most 
\[\tilde O\left(f^{-1/k} n^{2 + 1/k} |\mathcal H|  + m |\mathcal H| f \right).\]
\end{lemma}
\begin{proof}
We proceed as in Lemma~\ref{lem:randomized-running-time} by first analyzing the preprocessing, and assuming without loss of generality that $m \geq f^{1-1/k} n^{1+1/k} \geq fn$.  

Our set system has $O(|\mathcal H| f^2)$ sets.  For each hash function in $\mathcal H$, we can create the $f$ buckets in $\tilde O(n)$ time (by evaluating the hash function on each vertex in $\tilde O(1)$ time).  Then we can build each of the $\Theta(f^2)$ sets for that function in time $O(n/f)$, so the time to create all of the $\Theta(|\mathcal H| f^2)$ sets is $\tilde O(|\mathcal H| f n)$.  Creating the layered graphs takes additional $O(|\mathcal H| f^2 (n/f) k) = O(|\mathcal H| n f k) = O(mk |\mathcal H|)$ time.  Initializing all $O(|\mathcal H| f^2)$ data structures from Theorem~\ref{thm:dynamic-connectivity} and inserting the initial edges takes time $O(|\mathcal H| f^2 \cdot (kn/f)^2) = O(k^2 n^2 |\mathcal H|) = \tilde O(f^{-1/k} n^{2+1/k} |\mathcal H|)$.  
As in the analysis of the fast randomized algorithm (Lemma~\ref{lem:randomized-running-time}), while creating the sets we can record for each vertex a sorted list of which sets it is in, and then can create each $L_e$ set by simple set intersection.  Since each vertex is in $O(|\mathcal H| f)$ sets, this takes time $O(m |\mathcal H| f)$.   

Thus our total preprocessing time is $\tilde O(f^{-1/k} n^{2+1/k} |\mathcal H| + m |\mathcal H| f)$.

We now analyze the main greedy loop.  For every $e = (u,v) \in E$, the algorithm performs a connectivity query in $|L_e| \leq O(|\mathcal H|)$ of the layered subgraphs.  By Theorem~\ref{thm:dynamic-connectivity}, this takes $O(m |\mathcal H|)$ total time.  When we decide to add an edge $e$ to the spanner (which happens at most $O(f^{1-1/k} n^{1+1/k})$ times by Lemma~\ref{lem:deterministic-size}), we have to do $O(k)$ insertions into each of the $O(|L_e|) = O(|\mathcal H|)$ layered graphs in $L_e$.  The amortized cost of each insertion is $O(kn/f)$ by Theorem~\ref{thm:dynamic-connectivity} (since the number of nodes in each layered graph is $O(kn/f)$ by Lemma~\ref{lem:FS-deterministic}), and hence the total time of all insertions is $O(f^{1-1/k} n^{1+1/k} (kn/f) k |\mathcal H|) = O(k^2 f^{-1/k} n^{2+1/k} |\mathcal H|)$.  This is asymptotically larger than $m |\mathcal H|$ since $m < n^2$ and $f < n$, and hence the running time of the main loop is $O(k^2 f^{-1/k} n^{2+1/k} |\mathcal H|)$.
\end{proof}

This now finally allows us to prove Theorem~\ref{thm:deterministic-main}, which we restate here for clarity.
\begin{customthm}{\ref{thm:deterministic-main}}
There is a deterministic algorithm which constructs an $f$-VFT $(2k-1)$ spanner with at most $O(f^{1-1/k} n^{1+1/k})$ edges in time $\tilde O(f^{4-1/k} n^{2+1/k} + mf^5)$, and if $f = \poly(n)$ (i.e., $f \geq n^c$ for some constant $c > 0$) then the running time improves to $\tilde O(f^{1-1/k} n^{2+1/k} + mf^2)$ 
\end{customthm}
\begin{proof}
The fault tolerance and size bounds are from Lemmas~\ref{lem:deterministic-correct} and \ref{lem:deterministic-size}.  Using Lemma~\ref{lem:deterministic-running-time} with the almost universal construction of Theorem~\ref{thm:AGHP} (with $\delta = 2$) gives a running time of $\tilde O(k^2 f^{-1/k} n^{2+1/k} f^4 + mf^5) = \tilde O(f^{4-1/k} n^{2+1/k} + mf^5)$.  

On the other hand, if $f \geq n^c$ for some constant $c > 0$ then we want to use the almost universal construction of Theorem~\ref{thm:high-f-universal}.  But as discussed earlier, we have to be a little careful since if we use this construction then the range is a function of $\delta$ but $\delta$ is also a function of the range.  So we need to show that we can set $\delta$ so that it is a constant and the range is $[4\delta f]$ (since that is the construction we used).  We first set $\delta = \frac{\log n}{\log f} \leq 1/c$ and set the range of the hash family to be $4 \delta f$.  Then if we use Theorem~\ref{thm:high-f-universal} with this range, the theorem implies that this is a $\delta' = \frac{\log n}{\log (4 \delta f)}$-almost universal family.  But $\delta'$ is clearly at most $\delta$, and hence it is also a $\delta$-almost universal family.  So we have, as required, a $\delta$-almost universal family with range $[4 \delta f]$ where $\delta \leq 1/c$ is a constant.  When we use this family, Lemma~\ref{lem:deterministic-running-time} implies that the running time is $\tilde O(k^2 f^{-1/k} n^{2+1/k} f + mf^2) = \tilde O(f^{1-1/k} n^{2+1/k} + mf^2)$.
\end{proof}

\section{Conclusion and Open Questions} \label{sec:conclusion}
In this paper we gave the first polynomial-time algorithm to construct optimal-size vertex fault-tolerant spanners.  Our algorithm, after being optimized for running time, is also significantly faster than the previous best (non-optimal) polynomial time algorithm.  We also derandomized our algorithm to get a deterministic algorithm, which is always polynomial time and, in the most interesting regime of $f$ being polynomial in $n$, is just as fast (ignoring polylogarithmic factors) as our randomized algorithm.

There are still a number of tantalizing open problems involving fault tolerant spanners.  Algorithmically, while we significantly optimized the running time to make it surprisingly efficient, even in the regime $f=O(1)$ it is still not as fast the fastest algorithms for non-FT spanners (e.g., \cite{BaswanaS:07} which has running time $\tilde O(m)$).  These fast non-FT spanner algorithms are \emph{not} the greedy algorithm, which seems to be unable to achieve such an efficient running time.  Is it possible to compute optimal-size fault-tolerant spanners in time $\tilde O(m)$?  Similarly, there has been significant work on computing spanners in other models of computation, most notably in distributed and parallel models (see~\cite{BaswanaS:07,DGP07,DGP07,ParterY18,biswas2020massively} for a small sampling of such results).  The greedy algorithm is typically difficult to parallelize or implement efficiently distributedly (particularly in the presence of congestion), so there is the obvious question of computing optimal-size fault tolerant spanners efficiently in these models.

Next, we note that while all of our constructions and bounds work as well in the case of edge faults, they are not known to be optimal.
The best-known lower bound on the size of an $f$-edge fault tolerant $(2k-1)$ spanner, proved in~\cite{BDPW18}, is only $\Omega\left(f^{\frac12 (1-1/k)} n^{1+1/k}\right)$ rather than $\Omega\left(f^{1-1/k} n^{1+1/k}\right)$ as for vertex faults (for $k \ge 3$; for $k=2$ the lower bounds are both $\Omega\left(f^{1/2} n^{3/2}\right)$, and hence the size bounds achieved here are optimal).
Closing this gap for general $k$, by either giving improved upper bounds or improved lower bounds (or both), is probably the most important open question about fault-tolerant spanners.

Finally, there are some additional nice properties of the non-faulty greedy algorithm, and it would be interesting to determine whether these have desirable analogs for FT (slack-)greedy algorithms as well.
In particular: the non-FT greedy algorithm gives optimal spanners for several important graph classes like Euclidean graphs and doubling metrics \cite{LS19}, it produces optimal spanners as measured by \emph{lightness} \cite{FS20}, and there is experimental evidence that it performs particularly well on graphs encountered in practice \cite{FG05}.

\section*{Acknowledgements} We would like to thank Xin Li for many helpful discussions about derandomization, and in particular for pointing us towards message authentication codes.

\bibliographystyle{alpha}
\bibliography{refs}

\end{document}